\documentclass[12pt]{article}

\usepackage{amsthm}
\usepackage{natbib}
\usepackage{amsmath}
\usepackage{amssymb}
\usepackage{tikz}
\usepackage{bbm}
\usepackage{color}

\allowdisplaybreaks

\newcommand{\set}{\mathbbm}
\newcommand{\OO}{\mathrm{O}}
\newcommand{\e}{\mathrm{e}}
\def\deg{\operatorname{deg}}
\def\stirlingii#1#2{S_2(#1,#2)}
\def\edit#1{#1}

\newtheorem{example}{Example}
\newtheorem{definition}[example]{Definition}
\newtheorem{theorem}[example]{Theorem}
\newtheorem{lemma}[example]{Lemma}

\newtheorem{proposition}[example]{Proposition}
\newtheorem{remarks}[example]{Remarks}

\title{On the length of integers in telescopers for proper hypergeometric terms}

\author{Manuel Kauers\thanks{Supported by FWF project Y464-N18.}\\
   Research Institute for Symbolic Computation\\ 
  Johannes Kepler University\\
  Linz, Austria
  \and
   Lily Yen\\
   Capilano University and\\ 
  Simon Fraser University\\ 
  Vancouver, Canada
}

\begin{document}
\maketitle
\begin{abstract}
  We show that the number of digits in the integers of a creative telescoping
  relation of expected minimal order for a bivariate proper hypergeometric term has
  essentially cubic growth with the problem size. For telescopers of higher
  order but lower degree we obtain a quintic bound. Experiments suggest that
  these bounds are tight. As applications of our results, we give an improved
  bound on the maximal possible integer root of the leading coefficient of a
  telescoper, and the first discussion of the bit complexity of creative
  telescoping.
\end{abstract}

\section{Introduction}

Creative telescoping is a backbone of symbolic summation. It permits the construction of  
recurrence equations for definite sums. In its classical version,
it is applied 
 to sums whose summands are hypergeometric terms.  This situation was
intensively studied during the 1990s (see \citet{petkovsek97} and the references given there
for an overview on the classical results). 
 While during the first decade of this
century most research in the area 
 focussed on generalizing creative
telescoping to sums whose summands are more complicated
(see, for instance, the survey articles of \citet{koutschan13} and \citet{schneider13} and the references
given there), 
 the hypergeometric case is recently getting back into the focus. 
There is now a general interest in
getting a better understanding of the sizes of the output of summation
algorithms, and of the amount of time spent on the computation. 
\edit{First complexity estimates for summation (and integration) algorithms were given
by \cite{takayama95} and~\cite{gerhard04}. More recent works include the articles 
by \cite{bostan10} and~\cite{chen12c,chen12}. In the present paper, we continue these
investigations.} We work out bounds for the length of the
integers that may appear in the output of creative telescoping algorithms,
complementing earlier results 
 given for the order and the degree of creative
telescoping relations. As corollaries of our bounds, we obtain a new bound on
the maximal integer root of the leading coefficient as well as a first bound on
the bit complexity of creative telescoping.

Throughout this article, we consider a proper hypergeometric term
\begin{equation}\label{eq:hgdef}
   h
   =
   p\, x^n y^k \!
   \prod_{m=1}^M
      \frac{\Gamma(a_m n+a'_m k+a''_m)\Gamma(b_mn-b'_m k+b''_m)}
           {\Gamma(u_mn+u'_mk+u''_m)\Gamma(v_mn-v'_mk+v''_m)}
   ,
\end{equation}
where $p\in\set Z[n,k]$, $M\in\set N$ is fixed, $x$, $y$, $a_m$, $a'_m$,
$a''_m$, $b_m$, $b'_m$, $b''_m$, $u_m$, $u'_m$, $u''_m$, $v_m$, $v'_m$, $v''_m$
are fixed nonnegative integers, and $n$ and~$k$ are variables. To avoid discussion of
degenerate cases, we assume throughout that $h$ is not a rational function.
The assumption that there are exactly $M$ Gamma-terms for each of the four
types is without loss of generality, because we can always add further
terms~$\Gamma(0n+0k+\edit{1})$ without changing~$h$.

A creative telescoping relation for $h$ is a pair $(L,C)$, where
$L=\ell_0+\ell_1 S_n+\cdots+\ell_r S_n^r\in\set Z[n][S_n]\setminus\{0\}$ is a
nonzero recurrence operator in~$n$, free of~$k$, and $C\in\set Q(n,k)$ a
bivariate rational function in $n$ and~$k$ (which may well be zero), with the
property
\[
  L(h) = (S_k-1)(Ch).
\]
The symbols $S_n$ and $S_k$ refer to the usual shift operators $n\leadsto n+1$,
$k\leadsto k+1$, respectively. The operator $L$ is called a \emph{telescoper}
for~$h$, and $C$ is called a \emph{certificate} for $L$ and~$h$.
\edit{
Note that with $h$ non-rational, and $C$ nonzero, $Ch$ is also non-rational, in particular, not constant. Therefore, $(S_k-1)(Ch)$ is nonzero. From the equality above, we thus have $L(h)$ nonzero, or specifically, $L$ nonzero. In short, when $h$ is non-rational,   we can be sure that every
nontrivial pair $(L,C)$ must have a nontrivial~$L$.}

If $h$ has finite support, i.e., \edit{for every $n\in\set N$ there are only finitely 
many $k$ with $h(n,k)\neq0$, and if $Ch$ is well-defined for all~$n,k$,}
then a telescoper $L$ annihilates the definite hypergeometric sum
$H(n):=\sum_k h(n,k)$. If not, a creative telescoping relation still gives rise
to an inhomogeneous recurrence for finite definite sums such as $\sum_{k=0}^n
h(n,k)$ or $\sum_{k=n}^{2n} h(n,k)$.  See~\citet{petkovsek97} for details. 

The classical Zeilberger
algorithm~\citep{zeilberger90a,zeilberger91,petkovsek97} finds a creative
telescoping relation for any given proper hypergeometric term~$h$. This
algorithm is based on Gosper's algorithm~\citep{gosper78} for indefinite
hypergeometric summation and delivers a creative telescoping relation $(L,C)$
for which the order~$r$ of $L$ is minimal. An alternative algorithm proposed by 
Apagodu and Zeilberger~(\citeyear{mohammed05}) does not use Gosper's algorithm
during the computation but only in its correctness proof. \edit{This algorithm
also finds creative telescoping relations for proper hypergeometric terms, but
unlike Zeilberger's original algorithm there is no guarantee that the telescoper
has minimal possible order.} The key observation
behind \edit{the algorithm of Apagodu and Zeilberger} is that $L=\ell_0+\ell_1S_n+\cdots+\ell_rS_n^r\in\set
Q[n][S_n]$ is a telescoper for~$h$ if there exists some polynomial $Y\in\set
Q[n,k]$ with the property
\begin{equation}\label{eq:2}
  \ell_0 P_0 + \cdots + \ell_r P_r = Q\,S_k(Y) - R\,Y,
\end{equation}
where 
\begin{alignat*}1
  P_i &= x^i S_n^i(p) \prod_{m=1}^M \Bigl(
         (a_m n +a_m'k + a_m'')^{\overline{ia_m}}
         (b_m n - b_m'k + b_m'')^{\overline{ib_m}}
\\
      &\qquad{}\times
         (u_m n + u_m' k + u_m'' + iu_m  )^{\overline{(r-i)u_m}}
         (v_m n - v_m' k + v_m'' + iv_m )^{\overline{(r-i)v_m}}\Bigr)
\\
      &\kern290pt(i=0,\dots,r),
\\
    Q &= y \prod_{m=1}^M 
      (a_m n +a_m'k + a_m'')^{\overline{a_m'}}
      (v_m n - v_m' k + v_m'' + r v_m - v_m' )^{\overline{v_m'}},
\\
    R &= \prod_{m=1}^M 
      (u_m n + u_m' k + u_m'' + ru_m - u_m' )^{\overline{u_m'}}
      (b_m n - b_m'k + b_m'')^{\overline{b_m'}}.
\end{alignat*}
\edit{Here and below we write $x^{\overline{m}}:=x(x+1)\cdots(x+m-1)$ and
$x^{\underline{m}}:=x(x-1)\cdots(x-m+1)$ to denote the rising and falling factorial, respectively.}
\edit{A certificate is then given by 
\[
  C = \frac Yp\prod_{m=1}^M \frac{(b_m n - b_m'k + b_m'')^{\overline{b_m'}} }
       { (u_mn+u_m'k+u_m'')^{\overline{ru_m-u_m'}} (v_mn-v_m'k+v_m'')^{\overline{rv_m}}},
\]
so that 
\begin{alignat}1\label{eq:Ch}
  Ch = Y x^n y^k \prod_{m=1}^M \frac{\Gamma(a_mn+a_m'k+a_m'')\Gamma(b_mn-b_m'k+b_m''+b_m')}
     {\Gamma(u_mn+u_m'k+u_m''+ru_m)\Gamma(v_mn-v_m'k+v_m''+rv_m)}.
\end{alignat}
}

\edit{These results are due to Apagodu and Zeilberger~(\citeyear{mohammed05}). 
For a justification of the formulas, see either their article, or, with the 
notation we are using here, the paper by~\cite{chen12c}.}  The following 
definition contains certain quantities in terms of which bounds on the size 
of the telescoper of~$h$ can be formulated.

\begin{definition}\label{def:greek}
  For a proper hypergeometric term $h$ as above, define
  \begin{align*}
   \nu &=
      \max\Bigl\{
         \sum_{m=1}^M (a'_m+v'_m),
         \sum_{m=1}^M (u'_m+b'_m)
      \Bigr\}, \quad
      &
   \delta &= \deg(p),\\
   \vartheta&= \max\Bigl\{
         \sum_{m=1}^M (a_m+b_m),
         \sum_{m=1}^M (u_m+v_m)
      \Bigr\},\quad
      &
   \lambda &=\sum_{m=1}^M(u_m+v_m),
\\ \mu &= \sum_{m=1}^M (a_m+b_m-u_m-v_m). 
\end{align*}
Furthermore, we let 
\begin{alignat*}1
  \Omega:=\max_{m=1}^M\max\{&|a_m|, |a'_m|, |a''_m|, |b_m|, |b'_m|, |b''_m|,
                             |u_m|, |u'_m|, |u''_m|, |v_m|, |v'_m|, |v''_m|\}
\end{alignat*} 
be a bound on the integers appearing in the arguments of the $\Gamma$ terms of~$h$.
\end{definition}

Apagodu and Zeilberger show that $h$ admits a telescoper of order~$r$ for every
$r\geq\nu$, or in other words, that if $r$ is the order of the minimal
telescoper, then $r\leq\nu$. Generically this bound is tight. 
\citet{chen12c} supplement this result with
information about the degrees of the coefficients of the telescoper. They show
that for every $r\geq\nu$ and every $d$ satisfying
\[
  d>\frac{(\vartheta\nu-1)r + \tfrac12\nu(2\delta+|\mu|+3-(1+|\mu|)\nu)-1}{r-\nu+1},
\]
there exists a telescoper $L=\ell_0+\cdots+\ell_rS_n^r$ with $\max_{i=0}^r
|\ell_i|\leq d$. The purpose of the present article is to refine the analysis
one step further by giving bounds on the length of the integers appearing in the
coefficients $\ell_i$ of a telescoper~$L$ of~$h$. In Theorem~\ref{thm:minimal} in
Section~\ref{sec:minimal}, we show that hypergeometric terms $h$ have a telescoper 
of order $r=\nu$ whose integer coefficients have no more than $\OO(\Omega^3\log(\Omega))$ digits. 
In Theorem~\ref{thm:nonminimal} in Section~\ref{sec:minimal}, we show furthermore
that there are telescopers of order $r=\OO(\Omega)$ and degree $d=\OO(\Omega^2)$
whose integer coefficients have no more than $\OO(\Omega^5\log(\Omega))$ digits.
For both estimates, we provide experimental data that \edit{indicate} 
 that our bounds 
are sharp. 

\section{Bounding Tools}

In order to bound the integers arising in the coefficients of a telescoper, we
need to know by how much the size of the integers can grow during the various steps
of the computation. In particular, we need to know how adding, multiplying,
and shifting of polynomials may affect the length of their coefficients, and how long
the integer coefficients can become in the solution of a system of linear 
equations with polynomial coefficients. In this section we provide a
collection of results in this direction.

The coefficient length of a polynomial depends on the basis with respect to which
the polynomial is expressed. We are mostly interested in the coefficient length
with respect to the standard monomial basis $1,x,x^2,x^3,\dots$, but we will also
have occasion to use alternative bases. In the following definition we introduce
the notational distinction which will be used below. 

\begin{definition}
  \begin{enumerate}
  \item For $p=\sum_{i=0}^d p_i n^i\in\set Q[n]$, we call
    $|p|:=|p|_s:=\max_{i=0}^d |p_i|$ the \emph{(standard) height} or the
    \emph{(standard) norm} of~$p$.
  \item For $p=\sum_{i=0}^d p_i \binom{n}{i}\in\set Q[n]$, we call $|p|_b:=\max_{i=0}^d
    |p_i|$ the \emph{binomial height} or the \emph{binomial norm} of~$p$.
  \item For $p=\sum_{i=0}^d\sum_{j=0}^ep_{i,j}n^ik^j\in\set Q[n,k]$,
    we define $\|p\|_{s,s}:=\max_{i=0}^d\max_{j=0}^e|p_{i,j}|$.
  \item For $p=\sum_{i=0}^d\sum_{j=0}^ep_{i,j}n^i\binom kj\in\set Q[n,k]$,
    we define $\|p\|_{s,b}:=\max_{i=0}^d\max_{j=0}^e|p_{i,j}|$.
  \end{enumerate}
\end{definition}

Note that $|\cdot|_s$, $|\cdot|_b$, $\|\cdot\|_{s,s}$, and $\|\cdot\|_{s,b}$ are
indeed norms, i.e., they satisfy absolute scalability, triangle inequality, and
they are zero only when the argument is zero. The following lemmas give bounds
for shifted polynomials, for products of polynomials, and, to begin with, a
connection between the standard norm and the binomial norm.

\begin{lemma}[Conversion]\label{lemma:convert}
For all $p \in \set Q[n,k]$, we have $\|p\|_{s,b} \le \deg_k(p)!^2\, \lVert p\rVert_{s,s}$.
\end{lemma}
\begin{proof}
Recall from Equation (6.10) on page 262 of~\citet{GKP94}: 
\[
k^m = \sum_{i \ge 0} \stirlingii{m}{i}k^{\underline{i}} 
	=\sum_{i \ge 0} \stirlingii{m}{i}i! \frac{k^{\underline{i}}}{i!} 
	=\sum_{i \ge 0} \stirlingii{m}{i}i! \binom{k}{i},
\]
where $\stirlingii{m}{i}$ is the Stirling number of the second kind.
Write $p=p_0+p_1k+\cdots+p_dk^d$ with $p_0,\dots,p_d\in\set Q[n]$. Then 
\[
p = 
\sum_{j=0}^d p_j k^j = \sum_{j=0}^d 
\biggl( p_j \sum_{i=0}^j \stirlingii{j}{i} i! \binom{k}{i} \biggr)
=\sum_{i=0}^d \biggl(\sum_{j=i}^d p_j \stirlingii{j}{i} i!\biggr) \binom{k}{i}.
\]
Thus, for the binomial height of $p$, we find
\begin{alignat*}1
\|p\|_{s,b} 
&=\max_{i=0}^d \left| \sum_{j=i}^d p_j \stirlingii{j}{i} i! \right|
\leq \max_{i=0}^d \sum_{j=i}^d |p_j| \stirlingii{j}{i} i!
\leq \max_{i=0}^d \sum_{j=i}^d \|p\|_{s,s} \stirlingii ji d!\\
&
\leq \|p\|_{s,s}\, d! \sum_{j=0}^d \stirlingii{d}{i}
\le \|p\|_{s,s}\, d!\,\mathbf{B}_d
\le \|p\|_{s,s}\, d!^2,
\end{alignat*}
where $\mathbf{B}_d$ denotes the $d$th Bell number. 
\end{proof}

\begin{lemma}[Shift]\label{lemma:shift}
  For $q\in\set Q[n,k]$ and $r\in\set N$, we have
  $\|S_n^r(q)\|_{s,s}\leq(1+r)^{\deg_n(q)}\lVert q\rVert_{s,s}$ and
  $\|S_n^r(q)\|_{s,b}\leq(1+r)^{\deg_n(q)}\lVert q\rVert_{s,b}$.
\end{lemma}
\begin{proof}
  For $\|\cdot\|_{s,s}$, this is Lemma 3.4 of~\citet{yen96}. 
  It then also holds for $|\cdot|_s$ and polynomials in $\set Q[n]\subseteq\set Q[n,k]$.
  If finally $q=\sum_{i=0}^d q_i\binom ki\in\set Q[n,k]$ for certain $q_i\in\set Q[n]$, 
  then $\|S_n^r(q)\|_{s,b}=\max_{i=0}^d|S_n^r(q_i)|_s
  \leq (1+r)^d \max_{i=0}^d |q_i|_s=(1+r)^d \lVert q\rVert_{s,b}$, so it also holds for the
  norm $\|\cdot\|_{s,b}$.
\end{proof}

\begin{lemma}[Product]\label{lemma:product}
 \begin{enumerate}
 \item\label{lemma:product:3} 
   For $p_1,\dots,p_m\in\set Q[n]$, we have
   \[
     \biggl|\prod_{i=1}^m p_i\biggl|\leq \bigl(\max_{i=1}^m\deg(p_i)+1\bigr)^{m-1}\prod_{i=1}^m |p_i|.
   \]
 \item\label{lemma:product:4} Let $p_1,p_2,\dots,p_m\in\set Q[n,k]$ be polynomials of total degree~1, and
   $M\in\set N$ be such that $\|p_i\|_{s,s}\leq M$ for $i=1\dots,m$. Then
   for every $q\in\set Q[n,k]$, we have 
   \[
     \|p_1p_2\cdots p_mq\|_{s,b}\leq (2M)^m(\deg_k(q)+2)^{\overline{m}}\lVert q\rVert_{s,b}.
   \]
 \end{enumerate}
\end{lemma}
\begin{proof}
  \begin{enumerate}
  \item It suffices to prove the case $m=2$. The general case then follows
    immediately by induction. To show the case $m=2$, consider two polynomials
    $p=\sum_{i=0}^d p_i n^i$ and $q=\sum_{i=0}^e q_i n^i$. The coefficient of
    $n^j$ in $pq$ is $\sum_{i=0}^{d+e} p_i q_{j-i}$, where we understand
    coefficients as being zero if their index is out of range. For every $j$,
    the sum can have at most $\min\{\deg(p),\deg(q)\}+1$ nonzero terms, and as
    each term is bounded by $|p_i q_{j-i}|\leq|p|\,|q|$, the claim follows.
  \item It suffices to prove the case $m=1$. The general case then follows
    immediately by induction on~$m$.  Consider $p=a+bk+cn\in\set Q[n,k]$ and
    write $q=\sum_{i=0}^d q_i\binom ki$ with $q_0,\dots,q_d\in\set Q[n]$. 
    Observe that
\[
  (u k + v) \binom k i = (u i + v)\binom k i + u(i + 1)\binom k{i+1}.
\]
Therefore 
\begin{alignat*}1
  pq&=(a n + b k + c) \sum_{i=0}^d q_i\binom ki\\
    &=\sum_{i=0}^{d} \bigl(b q_i k + (a n + c)q_i\bigr)\binom ki\\
    &=\sum_{i=0}^{d} (a n + b i + c)q_i\binom k i + b (i+1)q_i\binom k{i+1}\\
    &=\sum_{i=0}^{d+1} \bigl((a n + b i + c)q_i + b iq_{i-1}\bigr)\binom ki.
\end{alignat*}
Because of
\[
  |(a n + b i + c)q_i + b iq_{i-1}|\leq 2(i+1)M\max\{|q_i|,|q_{i-1}|\},
\]
and  $|q_i|\leq\|q\|_{s,b}$ for all~$i$, it follows that
\[
 \|p\,q\|_{s,b}\leq \max_{i=0}^{d+1} 2(i+1)M\max\{|q_i|,|q_{i-1}|\}\leq 2\,(d+2) M\,\lVert q\rVert_{s,b}
\]
as claimed. 
  \end{enumerate}
\end{proof}

Finally, we need a bound on the length of the integers which may appear in the
basis vectors of the nullspace of a matrix with univariate polynomial
entries. The result below takes into account that the columns of the matrix may
be split into two groups for which different bounds on the degrees and \edit{heights}
are known. Although matrices and vectors all have coefficients in $\set
Z[x]$, all linear algebra notions (rank, kernel, linear independence, etc.) are
understood with respect to the ground field~$\set Q(x)$.

\begin{lemma}
\label{prop: matrix degree/height bound}
Let $A=(A_0,A_1)\in\set Z[x]^{n\times(m_0+m_1)}$ be a matrix of rank~$\rho$.
For $i=0,1$, let $d_i$ and $M_i$ be bounds on the degrees and standard heights of
the entries of $A_i\in\set Z[x]^{n\times m_i}$.  Assume that $A_0$ has full rank. Then $\ker A$
has a basis of vectors from $\set Z[x]^{m_0+m_1}$ that are bounded in degree by
$(m_0-1)d_0 + (\rho-m_0)d_1+\max\{d_0,d_1\}$ and in height by
\[
   \rho! (\max\{d_0,d_1\}+1)^{\rho-1}
      M_0^{m_0-1} M_1^{\rho-m_0} \max\{M_0,M_1\} .
\]
\end{lemma}
\begin{proof}
By selecting a maximal linearly independent set of rows from~$A$,
we may assume without loss of generality that $n=\rho$.
Furthermore, because $A_0$ has full rank, we have $\rho \geq m_0$,
and by exchanging columns within $A_1$ if necessary, we may
assume that $A_1=(W,V)$ where
$W \in \set{Z}[x]^{\rho\times (\rho-m_0)}$,
$V \in \set Z[x]^{\rho\times{(m_1-(\rho-m_0))}}$
and
$U := (A_0,W) \in \set Z[x]^{\rho \times \rho}$
has full rank.

A basis of $\ker A$ is given by the vectors
$(v_i,-e_i)\in\set Q(x)^{\rho+(m_0+m_1- \rho)}$
where $e_i\in\set Q(x)^{m_0+m_1- \rho}$
is the $i$th unit vector and
$v_i\in\set Q(x)^\rho$
is the unique solution of the inhomogeneous linear system
$U v_i = V e_i$.
The right hand side is of course just the $i$th column of~$V$.
According to Cramer's rule, the $j$th component of $v_i$ is
given by $\frac{\det U'}{\det U}$ where $U'$ is the matrix
obtained from $U$ by replacing its $j$th column by the $i$th
column of~$V$.
Multiplying all the basis vectors by $\det U$ gives a basis
of polynomial entries with integer coefficients.
By Lemma~\ref{lemma:product}.(\ref{lemma:product:3}),
and from the definition of the determinant,
\[
  \det ((a_{i,j}))_{i,j=1}^n 
  = \sum_{\pi\in S_n} \operatorname{sgn}(\pi)\prod_{i=1}^n a_{i,\pi(i)},
\]
the heights of the determinants $\det U'$ corresponding to columns
$j \le m_0$ are bounded by 
$\rho! (\max\{d_0,d_1\}+1)^{\rho-1} M_0^{m_0-1} M_1^{\rho-m_0 + 1}$;
and by
\(
   \rho! (\max\{d_0,d_1\}+1)^{\rho-1}
   M_0^{m_0} M_1^{\rho-m_0 }
\)
for $j > m_0$.
Combining both cases gives the claimed bound. 
The degree estimate  follows from the defining formula for the determinant
by the same reasoning.
\end{proof}

\section{Bounds for $P_0,\dots,P_r$, $Q$, and~$R$}

\edit{In Sections \ref{sec:minimal} and~\ref{sec:nonminimal}}
we will obtain our bounds on the height of the telescoper by making an ansatz for
$\ell_0,\dots,\ell_r$ and the coefficients of the polynomial $Y$ in equation~\eqref{eq:2},
comparing coefficients, and applying Lemma~\ref{prop: matrix degree/height bound} to
the linear system obtained from comparing coefficients in~\eqref{eq:2}. For doing so,
we need to determine the heights and degrees of the polynomials in this equation. 

For the degrees, we have $\deg(P_i)\leq\delta+r\vartheta$ and
$\deg(Q),\deg(R)\leq\nu$ by Lemmas 2 and~4 of~\citet{chen12c}, where $\deg$
refers to the total degree.

For the heights, we apply the lemmas of the previous section. Noting that
the products over the rising factorials consist of linear factors all of which
have heights bounded by $(r+2)\Omega-1$, it follows that
\begin{alignat*}3
  \|P_i\|_{s,b}
    &\leq (2(r+2)\Omega-2)^{r \lambda + i\mu}(\delta+2)^{\overline{r\lambda + i\mu}}\lVert x^i S_n^i(p)\rVert_{s,b}
    &\quad&\text{by Lemma~\ref{lemma:product}.(\ref{lemma:product:4})}\\
    &\leq (2(r+2)\Omega-2)^{\vartheta r}(\delta+2)^{\overline{\vartheta r}}\lVert x^i S_n^i(p)\rVert_{s,b}
    &\quad&\text{because $r\lambda+i\mu\leq\vartheta r$}\\
    &\leq |x|^i (\delta+\vartheta r+1)!(2(r+2)\Omega-2)^{\vartheta r}(1+i)^{\deg_n(p)}\lVert p\rVert_{s,b} 
    &\quad&\text{by Lemma \ref{lemma:shift}}\\
    &\leq \|p\|_{s,s} \delta!^2 (1+r)^\delta|x|^r(\delta+\vartheta r+1)!(2(r+2)\Omega-2)^{\vartheta r}
    &\quad&\text{by Lemma \ref{lemma:convert}}
\end{alignat*}
for every $i=0,\dots,r$. Note that the right hand side does not depend on~$i$ but only on~$r$
and quantities that are determined by the hypergeometric term~$h$. 

For $Y_j=\binom kj$, we have $S_k(Y_j)=\binom{k+1}j=\binom kj+\binom k{j+1}$;  therefore,
$\|S_k(Y_j)\|_{s,b}=\|Y_j\|_{s,b}=1$. Hence, since also the linear factors in the rising factorials in $Q$ and $R$
are all bounded in height by $(r+2)\Omega-1$, we obtain, again by using 
Lemma~\ref{lemma:product}.(\ref{lemma:product:4}),
\begin{alignat*}3
  \|Q\,S_k(Y_j)\|_{s,b}
   &\leq |y| (2(r+2)\Omega-2)^{\sum_{m=1}^M(a_m'+v_m')}(j+2)^{\overline{\sum_{m=1}^M(a_m'+v_m')}}\lVert S_k(Y_j)\rVert_{s,b}
   \\
   &\leq |y| (j+\nu+1)^\nu (2(r+2)\Omega-2)^\nu,
\end{alignat*}
and likewise
\[
  \|R\,Y_j\|_{s,b}\leq (j+\nu+1)^\nu (2(r+2)\Omega-2)^\nu
\]
for every $j\in\set N$. 

\section{The minimal telescoper}\label{sec:minimal}

Choose $r=\nu$ and $s=\delta+(\vartheta-1)\nu$, and make an ansatz 
\[
  Y=y_0+y_1\binom k1 + \cdots + y_s\binom ks
\]
with undetermined coefficients $y_0,\dots,y_s$.
Then, comparing like coefficients of $\binom kj$ in the equation
\[
  \ell_0 P_0 + \cdots + \ell_r P_r = Q\,S_k(Y) - R\,Y
\]
leads to a system of homogeneous linear equations with
$(r+1)+(s+1)=\delta+\vartheta\nu+2$ variables
$\ell_0,\dots,\ell_r,y_0,\dots,y_s$ and no more than
\[
 \max\bigl\{1+\max_{i=0}^r \deg_k(P_i),\ 1+\deg_k(Q)+s,\ 1+\deg_k(R)+s\bigr\}\leq\delta+\vartheta\nu+1
\] 
equations. This system obviously has a nontrivial solution.

The matrix $A\in\set
Z[n]^{(\delta+\vartheta\nu+1)\times(\delta+\vartheta\nu+2)}$ encoding this
system has the form $A=(A_L,A_C)$ where $A_L\in\set
Z[n]^{(\delta+\vartheta\nu+1)\times(\nu+1)}$ consists of the columns corresponding
to the variables $\ell_j$ in the telescoper part, and $A_C\in\set
Z[n]^{(\delta+\vartheta\nu+1)\times(\delta+(\vartheta-1)\nu+1)}$ consists of the
columns corresponding to the variables $y_j$ in the certificate part.  More
precisely, the entry of $A_L$ in row~$i$ and column~$j$ is the coefficient of
$\binom k{i-1}$ in $P_{j-1}$ ($i=1,\dots,\delta+\vartheta\nu+1$;
$j=1,\dots,\nu+1$), and the entry of $A_C$ in row~$i$ and column~$j$ is the
coefficient of $\binom k{i-1}$ in $Q\,S_k(\binom k{j-1})-R\,\binom k{j-1}$ 
($i=1,\dots,\delta+\vartheta\nu+1$; $j=1,\dots,\delta+(\vartheta-1)\nu+1$).

By the results of the previous section, the entries of $A_L$ have degree at most
$\delta+\vartheta\nu$ and height at most $\|p\|_{s,s} \delta!^2
(\nu+1)^\delta|x|^\nu(\delta+\vartheta\nu+1)!(2(\nu+2)\Omega-2)^{\vartheta\nu}$,
and the entries of $A_C$ have degree at most $\delta+\vartheta\nu$ and height at
most $(|y|+1)(\delta+\vartheta\nu+1)^\nu (2(\nu+2)\Omega-2)^\nu$.

We want to determine the height of the polynomials in the solution vectors
of~$A$.  There are two cases to distinguish. If $A_L$ has full rank, then we can
apply Lemma~\ref{prop: matrix degree/height bound}  with $A_0=A_L$,
$A_1=A_C$, $\rho\leq n=\delta+\vartheta\nu+1$, $m_0=\nu+1$,
$m_1=\delta+(\vartheta-1)\nu+1$. It implies the existence of a solution
$(\ell_0,\dots,\ell_\nu,y_0,\dots,y_{\delta+(\vartheta-1)\nu})\in\set
Z[n]^{(\nu+1)+(\delta+(\vartheta-1)\nu+1)}$ with
\begin{align*}
 |\ell_i|
   &\leq (\delta+\vartheta\nu+1)!(\max\{\delta+\vartheta\nu,\delta+\vartheta\nu\}+1)^{\delta + \vartheta\nu}
   \vphantom{\Big\}}
   \\
   &\quad\times\Bigl(
   \|p\|_{s,s} \delta!^2(\nu+1)^\delta|x|^\nu(\delta+\vartheta\nu+1)!(2(\nu+2)\Omega-2)^{\vartheta\nu}
   \Bigr)^\nu
   \\
   &\quad\times\Bigl(
   (|y|+1)(\delta+\vartheta\nu+1)^\nu (2(\nu+2)\Omega-2)^\nu
   \Bigr)^{\delta+\vartheta\nu+1-\nu}
   \\
   &\quad\times\max\Bigl\{
     \|p\|_{s,s} \delta!^2(\nu+1)^\delta|x|^\nu(\delta+\vartheta\nu+1)!(2(\nu+2)\Omega-2)^{\vartheta\nu},\\
   &\qquad\qquad\qquad
     (|y|+1)(\delta+\vartheta\nu+1)^\nu (2(\nu+2)\Omega-2)^\nu
   \Bigr\}
   \\[5pt]
   &\leq ((\delta+\vartheta\nu+1)!)^2
      (\delta+\vartheta\nu+1)^{\delta + \vartheta\nu}\vphantom{\Big\}}
   \\
   &\quad\times
      \Bigl(
         \|p\|_{s,s} \delta!^2(\nu+1)^\delta
      |x|^\nu (\delta + \vartheta\nu + 1)! (2(\nu + 2)\Omega-2)^{\vartheta\nu}
      \Bigr)^\nu
   \\
   &\quad\times
      \Bigl(
         (|y|  + 1)\, (\delta +  \vartheta\nu + 1)^\nu
      (2(\nu+2)\Omega-2)^\nu
      \Bigr)^{\delta+(\vartheta-1)\nu+1}
   \\
   &\quad\times\lVert p\rVert_{s,s} \delta!^2(\nu+1)^\delta(\delta+\vartheta\nu+1)^\nu (2(\nu+2)\Omega-2)^{\vartheta\nu}
   \max\bigl\{|x|^\nu,|y|+1\bigr\}\vphantom{\Big\}}
   \\[5pt]
   &\le\vphantom{\Big\}}\max\bigl\{|x|^\nu,|y|+1\bigr\}
   \lVert p\rVert_{s,s}^{\nu+1}
    (\delta+\vartheta\nu+1)!^{\nu+2}
        (\nu+1)^{\delta(\nu+1)}(|y|+1)^{\delta+(\vartheta-1)\nu+1}
\\
  &\quad{}\times \delta!^{2(\nu+1)} |x|^{\nu^2}(\delta+\vartheta\nu+1)^{\delta+(\vartheta+\delta+2)\nu+(\vartheta-1)\nu^2}
(2(\nu+2)\Omega-2)^{(\delta+\vartheta+1)\nu+(2\vartheta-1)\nu^2}
\end{align*}
for $i=0,\dots,\nu$. 

If $A_L$ does not have full rank, then it has itself a nonempty kernel. In this
case, if $(\ell_0,\dots,\ell_\nu)$ is a nontrivial kernel element of~$A_L$, then
$(\ell_0,\dots,\ell_\nu,0,\dots,0)$ is a nontrivial kernel element
of~$A=(A_L,A_C)$. Therefore, in this case it suffices to estimate the height of
the polynomial entries in the kernel of~$A_L$. To this end, we use again
Lemma~\ref{prop: matrix degree/height bound}, this time taking $A_0$ to be some
nonzero column (w.l.o.g.\ the first), $A_1$ the remaining columns,
$n=\delta+\vartheta\nu+1$, $m_0=1$, $m_1=\nu$, $\rho\leq\nu-1$. Using for both $A_0$ and $A_1$ 
the degree and height estimates stated above for $A_L$, we get the
bound
\begin{alignat*}1
  |\ell_i|
   &\leq (\nu-1)!(\delta+\vartheta\nu+1)^{\nu-2}\Bigl(
   \|p\|_{s,s} \delta!^2(\nu+1)^\delta|x|^\nu(\delta+\vartheta\nu+1)!(2(\nu+2)\Omega-2)^{\vartheta\nu}
   \Bigr)^{\nu -1}
\end{alignat*}
for $i=0,\dots,r$. As this is always less than or equal to the bound obtained before
for the case when $A_L$ has full rank, we have completed the proof of the following theorem.
\edit{Recall from the remarks made in the introduction that the assumption of a
non-rational $h$ excludes the degenerate case that the telescoper may be zero.}

\begin{theorem}\label{thm:minimal}
  Let $h$ be a non-rational proper hypergeometric term as in \eqref{eq:hgdef}, and let 
  $\delta,\vartheta,\nu,\Omega$ be as in Definition~\ref{def:greek}. Then there exists
  a telescoper for $h$ of order~$r=\nu$ whose polynomial coefficients are bounded in height by
  \begin{alignat*}1
&\max\bigl\{|x|^\nu,|y|+1\bigr\}
   \lVert p\rVert_{s,s}^{\nu+1}
    (\delta+\vartheta\nu+1)!^{\nu+1}
        (\nu+1)^{\delta(\nu+1)}(|y|+1)^{\delta+(\vartheta-1)\nu+1}
\\
  &\quad{}\times \delta!^{2(\nu+1)} |x|^{\nu^2}(\delta+\vartheta\nu+1)^{\delta+(\vartheta+\delta+2)\nu+(\vartheta-1)\nu^2}
(2(\nu+2)\Omega-2)^{(\delta+\vartheta+1)\nu+(2\vartheta-1)\nu^2}.
  \end{alignat*}
\end{theorem}

\begin{remarks}
\begin{enumerate}
\item In general, a hypergeometric term~$h$ does not have any telescoper of
  order smaller than~$\nu$, so the theorem makes a statement about the integers
  appearing in the minimal order telescoper of a ``generic'' hypergeometric
  term~$h$. For hypergeometric terms which do possess a smaller telescoper, the
  theorem remains true as it stands, but does not say anything about the size of
  the integers in the minimal telescoper.
\item Lemma~\ref{prop: matrix degree/height bound} also yields the degree bound
  $(\delta+\nu\vartheta)(\delta+\nu\vartheta+1)=\OO(\Omega^4)$, which is worse
  than the degree bound $\OO(\Omega^3)$ given by \citet{chen12c}. In the generic case,
  when the minimal telescoper order is~$\nu$, the solution space of
  the linear system discussed above has dimension~1, so that at least in this 
  case there is a telescoper of degree $\OO(\Omega^3)$ and height as stated above.
  We do not know if this also applies to the degenerate case. 
\item Considering $\|p\|_{s,s}$, $\delta$, $|x|$, $|y|$, and $M$ as fixed, and
  noting that $\nu$ and $\vartheta$ are bounded by~$2M\Omega$, the bound of
  Theorem~\ref{thm:minimal} is equal to
  $\e^{64(M\Omega)^3\log(\Omega)+\OO(\Omega^3)}$ as $\Omega$ tends to infinity.
  Combined with the degree bound $\OO(\Omega^3)$ (when $\nu$ is minimal)
  or~$\OO(\Omega^4)$ (when it's not), it follows that there is a telescoper of
  order $r=\nu=\OO(\Omega)$ of bit size $\OO(\Omega^7\log(\Omega))$ or
  $\OO(\Omega^8\log(\Omega))$, respectively.
\item The choice of the binomial basis in the ansatz for $Y$ is motivated by the
  fact that with respect to this basis the shift does not increase the norm. In
  the standard basis we have $S_k(k^j)=(k+1)^j=\sum_i\binom ji k^i$, whose
  standard norm is $\binom{j}{\lfloor j/2\rfloor}\leq 2^j$. Using this (almost
  tight) bound in the argument above leads to a suboptimal bound of the form
  $\e^{\OO(\Omega^4\log(\Omega))}$.  Of course, the choice of the basis with
  respect to $k$ used in the ansatz for $Y$ does not have any effect on the
  output telescoper~$L$, which is free of~$k$ by construction.
\end{enumerate}
\end{remarks}

We conclude the section by a family of hypergeometric terms which gives evidence
that the bound of Theorem~\ref{thm:minimal} seems to be asymptotically accurate.

\begin{example}\label{ex:minimal}
  For $\Omega=1,2,3,\dots$ consider the proper hypergeometric term
  $h_\Omega=\frac{\Gamma(\Omega k)}{\Gamma(\Omega n-k)}$. We have 
  computed the minimal telescoper $L_\Omega$ of $h_\Omega$ for $\Omega=1,\dots,23$
  and determined the length of the integers appearing in them. Let $H_\Omega$
  be the logarithm of the maximum over the absolute values of all integers appearing in~$L_\Omega$.
  In Figure~\ref{fig:fit}, we plot the normalized values $\frac{H_\Omega}{\Omega^3}$ (bullets, \tikz[baseline=-1ex] \node[fill,circle, inner sep=1pt] {};)
  against the following least square fits, testing the four hypotheses 
  $H_\Omega=\Theta(\Omega^3\log(\Omega))$, $\Theta(\Omega^3)$, $\Theta(\Omega^2\log(\Omega))$,
  or $\Theta(\Omega^2)$, respectively: 
  \begin{enumerate}
  \item $\displaystyle\log(\Omega)\Bigl(1.43 + \frac{3.30}\Omega - \frac{1.66}{\Omega^2}\Bigr)$ (solid line, \tikz[baseline=-1ex] \draw[solid] (0,0)--(1,0);)
    \rule{0pt}{2em}
  \item $\displaystyle1\Bigl(5.06 - \frac{9.22}{\Omega} + \frac{4.23}{\Omega^2}\Bigr)$ (densely dashed, \tikz[baseline=-1ex] \draw[densely dashed] (0,0)--(1,0);)
  \item $\displaystyle\frac{\log(\Omega)}\Omega\Bigl(34.8 - \frac{167}{\Omega} + \frac{221}{\Omega^2}\Bigr)$ (loosely dashed, \tikz[baseline=-1ex] \draw[loosely dashed] (0,0)--(1,0);)
  \item $\displaystyle\frac1\Omega\Bigl(58.9 - \frac{182}{\Omega} + \frac{123}{\Omega^2}\Bigr)$ (dotted, \tikz[baseline=-1ex] \draw[dotted] (0,0)--(1,0);)
    \rule[-1.5em]{0pt}{1.5em}
  \end{enumerate}
  The best fit is given by the first hypothesis, suggesting that the bound
  proven above is asymptotically accurate. 
  

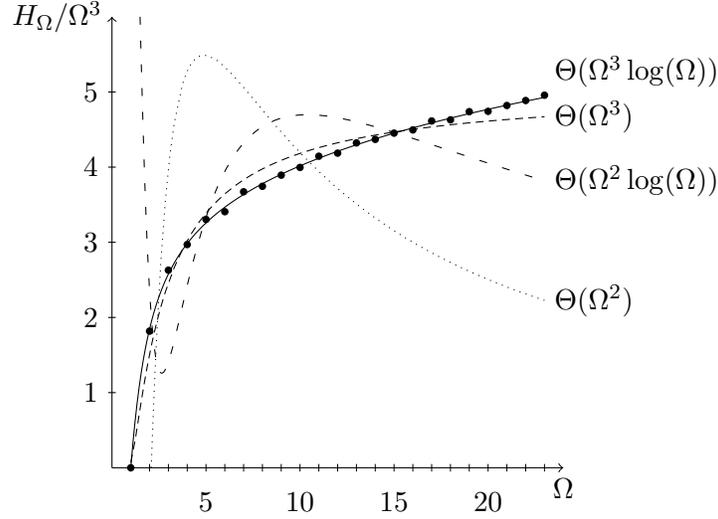
\begin{figure}[ht]
\centering\small
\begin{tikzpicture}[xscale=0.25]
   \draw[->] (0,0) -- (24,0) node[below]{$\Omega$};
   \draw[->] (0,0) -- (0,6) node[left]{$H_\Omega/\Omega^3$};
   \foreach \i in {1,2,...,23}
      \draw[very thin,yscale=0.25] (\i,0.2) -- +(0,-0.4);
   \foreach \i in {5,10,15,20}
      \node[below] at (\i,-0.2) {$\i$};
   \foreach \i in {1,...,5}
      \draw (0.2,\i) -- +(-0.4,0) node[left]{$\i$};
   \draw[solid] plot coordinates{ 
      (1.000000, 0.000000) (1.115667, 0.334043) (1.216309, 0.591211) 
      (1.329491, 0.846252) (1.443423, 1.071031) (1.556814, 1.267616) 
      (1.661941, 1.429352) (1.770795, 1.579132) (1.883371, 1.717899) 
      (1.995587, 1.842302) (2.111014, 1.957886) (2.212683, 2.050729) 
      (2.327137, 2.146549) (2.442062, 2.234694) (2.552813, 2.312970) 
      (2.653386, 2.379061) (2.772977, 2.452190) (2.874286, 2.510011) 
      (2.992132, 2.573038) (3.096446, 2.625425) (3.210895, 2.679609) 
      (3.319878, 2.728316) (3.433590, 2.776428) (3.538013, 2.818398) 
      (3.650647, 2.861509) (3.767642, 2.904125) (3.869487, 2.939581) 
      (3.979483, 2.976298) (4.093118, 3.012648) (4.204289, 3.046775) 
      (4.311852, 3.078547) (4.431283, 3.112495) (4.538596, 3.141892) 
      (4.653176, 3.172206) (4.757003, 3.198781) (4.870513, 3.226926) 
      (4.977317, 3.252597) (5.088958, 3.278641) (5.198108, 3.303372) 
      (5.312377, 3.328535) (5.422432, 3.352105) (5.534980, 3.375575) 
      (5.646595, 3.398251) (5.749158, 3.418591) (5.866706, 3.441347) 
      (5.971843, 3.461223) (6.083941, 3.481944) (6.191233, 3.501343) 
      (6.310367, 3.522411) (6.413528, 3.540272) (6.530617, 3.560133) 
      (6.637331, 3.577872) (6.754059, 3.596896) (6.854775, 3.613006) 
      (6.969480, 3.631024) (7.080279, 3.648110) (7.191005, 3.664883) 
      (7.301324, 3.681308) (7.407306, 3.696828) (7.521876, 3.713329) 
      (7.631069, 3.728800) (7.745989, 3.744823) (7.850020, 3.759105) 
      (7.964994, 3.774653) (8.075118, 3.789321) (8.184975, 3.803741)
      (8.299720, 3.818585) (8.405405, 3.832066) (8.513641, 3.845688) 
      (8.633146, 3.860520) (8.741348, 3.873766) (8.852005, 3.887138) 
      (8.964540, 3.900561) (9.067982, 3.912748) (9.178247, 3.925583) 
      (9.287674, 3.938167) (9.404726, 3.951462) (9.508019, 3.963057) 
      (9.627198, 3.976279) (9.734748, 3.988071) (9.841163, 3.999611) 
      (9.955593, 4.011883) (10.070487, 4.024065) (10.175023, 4.035031) 
      (10.286035, 4.046556) (10.394372, 4.057686) (10.511811, 4.069624) 
      (10.613428, 4.079849) (10.730139, 4.091477) (10.839854, 4.102297) 
      (10.948508, 4.112909) (11.057250, 4.123428) (11.168159, 4.134055) 
      (11.285721, 4.145210) (11.393513, 4.155342) (11.499919, 4.165253) 
      (11.613664, 4.175753) (11.727162, 4.186134) (11.829033, 4.195370) 
      (11.949706, 4.206215) (12.050680, 4.215210) (12.168781, 4.225642) 
      (12.281496, 4.235510) (12.382138, 4.244250) (12.495320, 4.254000) 
      (12.609252, 4.263732) (12.722643, 4.273337) (12.827770, 4.282172) 
      (12.936624, 4.291249) (13.049200, 4.300562) (13.161416, 4.309772) 
      (13.276843, 4.319170) (13.378512, 4.327387) (13.492967, 4.336568) 
      (13.607891, 4.345715) (13.718642, 4.354463) (13.819215, 4.362351) 
      (13.938806, 4.371663) (14.040115, 4.379494) (14.157961, 4.388540) 
      (14.262275, 4.396489) (14.376724, 4.405151) (14.485707, 4.413341) 
      (14.599419, 4.421826) (14.703842, 4.429566) (14.816476, 4.437859) 
      (14.933472, 4.446413) (15.035316, 4.453810) (15.145312, 4.461748) 
      (15.258948, 4.469895) (15.370118, 4.477812) (15.477681, 4.485423) 
      (15.597112, 4.493818) (15.704425, 4.501313) (15.819006, 4.509264) 
      (15.922832, 4.516425) (16.036342, 4.524205) (16.143146, 4.531481) 
      (16.254788, 4.539041) (16.363937, 4.546387) (16.478207, 4.554031) 
      (16.588262, 4.561348) (16.700809, 4.568786) (16.812424, 4.576119) 
      (16.914987, 4.582818) (17.032535, 4.590453) (17.137672, 4.597241) 
      (17.249770, 4.604438) (17.357062, 4.611288) (17.476196, 4.618849) 
      (17.579357, 4.625360) (17.696446, 4.632709) (17.803160, 4.639369) 
      (17.919888, 4.646613) (18.020604, 4.652830) (18.135309, 4.659874) 
      (18.246108, 4.666640) (18.356834, 4.673365) (18.467153, 4.680030) 
      (18.573135, 4.686399) (18.687705, 4.693249) (18.796898, 4.699742) 
      (18.911819, 4.706541) (19.015849, 4.712663) (19.130823, 4.719395) 
      (19.240947, 4.725810) (19.350804, 4.732177) (19.465549, 4.738793) 
      (19.571234, 4.744857) (19.679470, 4.751036) (19.798975, 4.757825) 
      (19.907177, 4.763940) (20.017834, 4.770163) (20.130369, 4.776462) 
      (20.233811, 4.782223) (20.344077, 4.788337) (20.453503, 4.794375) 
      (20.570555, 4.800802) (20.673848, 4.806447) (20.793027, 4.812929) 
      (20.900577, 4.818751) (21.006992, 4.824486) (21.121423, 4.830623) 
      (21.236316, 4.836756) (21.340852, 4.842311) (21.451864, 4.848184) 
      (21.560201, 4.853890) (21.677640, 4.860046) (21.779257, 4.865350) 
      (21.895968, 4.871414) (22.005683, 4.877088) (22.114338, 4.882684) 
      (22.223079, 4.888260) (22.333988, 4.893922) (22.451550, 4.899896) 
      (22.559342, 4.905351) (22.665748, 4.910712) (22.779493, 4.916419) 
      (22.892991, 4.922088) (23.000000, 4.927411) } 
      node[above right]{$\Theta(\Omega^3\log(\Omega))$};
   \draw[densely dashed] plot coordinates{ 
      (1.000000, .071331) (1.239769, .377225) 
      (1.479537, .763058) (1.688159, 1.085241) (1.896780, 1.377498) 
      (2.131397, 1.668039) (2.366014, 1.921522) (2.602185, 2.144284) 
      (2.838357, 2.339482) (3.308456, 2.662438) (3.744298, 2.902092) 
      (4.195587, 3.105530) (4.662310, 3.279803) (5.127537, 3.425507) 
      (5.606080, 3.552692) (6.027581, 3.649524) (6.502091, 3.744770) 
      (6.978549, 3.828378) (7.437705, 3.899538) (7.854663, 3.957433) 
      (8.350466, 4.019219) (8.770476, 4.066418) (9.259048, 4.116232) 
      (9.691517, 4.156356) (10.166003, 4.196647) (10.617827, 4.231826) 
      (11.089258, 4.265614) (11.522178, 4.294313) (11.989141, 4.323042) 
      (12.474184, 4.350696) (12.896416, 4.373141) (13.352438, 4.395846) 
      (13.823553, 4.417786) (14.284448, 4.437898) (14.730387, 4.456199) 
      (15.225526, 4.475303) (15.670430, 4.491472) (16.145461, 4.507781) 
      (16.575909, 4.521777) (17.046503, 4.536291) (17.489295, 4.549255) 
      (17.952140, 4.562141) (18.404656, 4.574129) (18.878398, 4.586079) 
      (19.334668, 4.597049) (19.801270, 4.607757) (20.264008, 4.617901) 
      (20.689216, 4.626832) (21.176553, 4.636636) (21.612433, 4.645039) 
      (22.077172, 4.653640) (22.521986, 4.661548) (23.000000, 4.669711) }
      node[right]{$\Theta(\Omega^3)$};
   \draw[loosely dashed] plot coordinates{ 
      (1.493187, 6.000000) (1.500118, 5.934260) (1.556814, 5.394931) 
      (1.609377, 4.911725) (1.661941, 4.455140) (1.716368, 4.016587) 
      (1.770795, 3.615969) (1.827083, 3.242546) (1.883371, 2.910116) 
      (1.995587, 2.362595) (2.111014, 1.942198) (2.212683, 1.673463) 
      (2.327137, 1.464812) (2.442062, 1.336506) (2.497438, 1.298664) 
      (2.552813, 1.274137) (2.577957, 1.266988) (2.603100, 1.262146) 
      (2.628243, 1.259486) (2.653386, 1.258889) (2.683284, 1.260708) 
      (2.713181, 1.265104) (2.743079, 1.271905) (2.772977, 1.280948) 
      (2.823631, 1.300963) (2.874286, 1.326290) (2.992132, 1.402375) 
      (3.096446, 1.485540) (3.210895, 1.589225) (3.319878, 1.696533) 
      (3.433590, 1.814406) (3.538013, 1.925873) (3.650647, 2.047765) 
      (3.767642, 2.174630) (3.869487, 2.284263) (3.979483, 2.401016) 
      (4.093118, 2.519156) (4.204289, 2.631781) (4.311852, 2.737618) 
      (4.431283, 2.851227) (4.538596, 2.949630) (4.653176, 3.050737) 
      (4.757003, 3.138776) (4.870513, 3.231123) (4.977317, 3.314312) 
      (5.088958, 3.397480) (5.198108, 3.475115) (5.312377, 3.552585) 
      (5.422432, 3.623615) (5.534980, 3.692724) (5.646595, 3.757850) 
      (5.749158, 3.814799) (5.866706, 3.876782) (5.971843, 3.929353) 
      (6.083941, 3.982539) (6.191233, 4.030781) (6.310367, 4.081422) 
      (6.413528, 4.122887) (6.530617, 4.167383) (6.637331, 4.205657) 
      (6.754059, 4.245143) (6.854775, 4.277301) (6.969480, 4.311861) 
      (7.080279, 4.343247) (7.191005, 4.372740) (7.301324, 4.400347) 
      (7.407306, 4.425269) (7.521876, 4.450527) (7.631069, 4.473045) 
      (7.745989, 4.495179) (7.850020, 4.513894) (7.964994, 4.533182) 
      (8.075118, 4.550344) (8.184975, 4.566243) (8.299720, 4.581608) 
      (8.405405, 4.594685) (8.513641, 4.607060) (8.633146, 4.619579) 
      (8.741348, 4.629926) (8.852005, 4.639582) (8.964540, 4.648483) 
      (9.067982, 4.655885) (9.178247, 4.662988) (9.287674, 4.669270) 
      (9.404726, 4.675182) (9.508019, 4.679737) (9.627198, 4.684256) 
      (9.734748, 4.687687) (9.841163, 4.690505) (9.955593, 4.692923) 
      (10.070487, 4.694741) (10.175023, 4.695890) (10.286035, 4.696605) 
      (10.394372, 4.696823) (10.511811, 4.696549) (10.613428, 4.695903) 
      (10.730139, 4.694713) (10.839854, 4.693177) (10.948508, 4.691277) 
      (11.057250, 4.689012) (11.168159, 4.686347) (11.285721, 4.683147) 
      (11.393513, 4.679889) (11.499919, 4.676383) (11.613664, 4.672332) 
      (11.727162, 4.667990) (11.829033, 4.663851) (11.949706, 4.658665) 
      (12.050680, 4.654100) (12.168781, 4.648512) (12.281496, 4.642941) 
      (12.382138, 4.637780) (12.495320, 4.631774) (12.609252, 4.625522) 
      (12.722643, 4.619106) (12.827770, 4.612992) (12.936624, 4.606502) 
      (13.049200, 4.599628) (13.161416, 4.592618) (13.276843, 4.585254) 
      (13.378512, 4.578643) (13.492967, 4.571069) (13.607891, 4.563330) 
      (13.718642, 4.555753) (13.819215, 4.548775) (13.938806, 4.540364) 
      (14.040115, 4.533146) (14.157961, 4.524651) (14.262275, 4.517046) 
      (14.376724, 4.508614) (14.485707, 4.500506) (14.599419, 4.491967) 
      (14.703842, 4.484060) (14.816476, 4.475463) (14.933472, 4.466463) 
      (15.035316, 4.458575) (15.145312, 4.450003) (15.258948, 4.441092) 
      (15.370118, 4.432325) (15.477681, 4.423798) (15.597112, 4.414284) 
      (15.704425, 4.405695) (15.819006, 4.396487) (15.922832, 4.388112) 
      (16.036342, 4.378925) (16.143146, 4.370253) (16.254788, 4.361162) 
      (16.363937, 4.352250) (16.478207, 4.342898) (16.588262, 4.333871) 
      (16.700809, 4.324622) (16.812424, 4.315435) (16.914987, 4.306980) 
      (17.032535, 4.297278) (17.137672, 4.288591) (17.249770, 4.279321) 
      (17.357062, 4.270443) (17.476196, 4.260579) (17.579357, 4.252034) 
      (17.696446, 4.242334) (17.803160, 4.233494) (17.919888, 4.223824) 
      (18.020604, 4.215483) (18.135309, 4.205986) (18.246108, 4.196818) 
      (18.356834, 4.187662) (18.467153, 4.178546) (18.573135, 4.169795) 
      (18.687705, 4.160345) (18.796898, 4.151348) (18.911819, 4.141891) 
      (19.015849, 4.133340) (19.130823, 4.123903) (19.240947, 4.114877) 
      (19.350804, 4.105887) (19.465549, 4.096513) (19.571234, 4.087893) 
      (19.679470, 4.079081) (19.798975, 4.069370) (19.907177, 4.060595) 
      (20.017834, 4.051638) (20.130369, 4.042549) (20.233811, 4.034212) 
      (20.344077, 4.025344) (20.453503, 4.016563) (20.570555, 4.007192) 
      (20.673848, 3.998942) (20.793027, 3.989447) (20.900577, 3.980901) 
      (21.006992, 3.972465) (21.121423, 3.963417) (21.236316, 3.954358) 
      (21.340852, 3.946137) (21.451864, 3.937430) (21.560201, 3.928956) 
      (21.677640, 3.919796) (21.779257, 3.911893) (21.895968, 3.902842) 
      (22.005683, 3.894358) (22.114338, 3.885982) (22.223079, 3.877623) 
      (22.333988, 3.869123) (22.451550, 3.860142) (22.559342, 3.851932) 
      (22.665748, 3.843853) (22.779493, 3.835244) (22.892991, 3.826681) 
      (23.000000, 3.818634) }
      node[right]{$\Theta(\Omega^2\log(\Omega))$};
   \draw[dotted] plot coordinates{ 
      (2.088313, 0.000000) (2.131397, .307060) (2.248705, 1.045601) 
      (2.366014, 1.693238) (2.602185, 2.755823) (2.838357, 3.553693) 
      (3.073407, 4.144946) (3.308456, 4.581388) (3.744298, 5.098616) 
      (3.969942, 5.260045) (4.195587, 5.369620) (4.428948, 5.440474) 
      (4.662310, 5.477722) (4.894923, 5.488855) (5.127537, 5.479803) 
      (5.606080, 5.415755) (6.027581, 5.325809) (6.502091, 5.202619) 
      (6.978549, 5.066127) (7.437705, 4.929098) (7.854663, 4.803478) 
      (8.350466, 4.655440) (8.770476, 4.532646) (9.259048, 4.393938) 
      (9.691517, 4.275415) (10.166003, 4.150305) (10.617827, 4.036081) 
      (11.089258, 3.921997) (11.522178, 3.821735) (11.989141, 3.718277) 
      (12.474184, 3.615769) (12.896416, 3.530467) (13.352438, 3.442258) 
      (13.823553, 3.355194) (14.284448, 3.273805) (14.730387, 3.198433) 
      (15.225526, 3.118419) (15.670430, 3.049638) (16.145461, 2.979272) 
      (16.575909, 2.918101) (17.046503, 2.853889) (17.489295, 2.795876) 
      (17.952140, 2.737593) (18.404656, 2.682816) (18.878398, 2.627681) 
      (19.334668, 2.576603) (19.801270, 2.526312) (20.264008, 2.478279) 
      (20.689216, 2.435675) (21.176553, 2.388563) (21.612433, 2.347902) 
      (22.077172, 2.306008) (22.521986, 2.267253) (23.000000, 2.226999) }
      node[right]{$\Theta(\Omega^2)$};
   \foreach \i/\j in {1/0.000, 2/1.818, 3/2.629, 4/2.970, 
         5/3.304, 6/3.407, 7/3.672, 8/3.745, 9/3.893, 
         10/3.996, 11/4.145, 12/4.186, 13/4.322, 14/4.369, 
         15/4.452, 16/4.495, 17/4.614, 18/4.629, 19/4.738, 
         20/4.742, 21/4.820, 22/4.886, 23/4.957}
      \node[fill,circle, inner sep=1pt] at (\i,\j) {};
\end{tikzpicture}
\caption{Heights of minimal telescopers}
\label{fig:fit}
\end{figure}
\edit{The corresponding comparison for the total bit size of the telescopers suggests 
that the bound $\Theta(\Omega^7\log(\Omega))$ is right. As the figure for this case 
looks very similar to the figure above, we do not reproduce it here.}
    
\end{example}

\section{Nonminimal telescopers}\label{sec:nonminimal}

As shown by \citet{chen12c}, telescopers of order $r>\nu$ may have much smaller
degrees than the (generically) minimal telescoper of order~$r=\nu$. More
precisely, the arithmetic size, i.e., the number of monomials $n^i S_n^j$ with a
nonzero coefficient appearing in a telescoper, which is bounded by $(r+1)(d+1)$,
is asymptotically smaller by one order of magnitude when $r=\alpha\nu$ for any
fixed constant~$\alpha>1$. It is therefore also interesting to bound the length
of the integers appearing in telescopers of nonminimal order.

In this section, we derive such a bound. Following \citet{chen12c}, we proceed by
analyzing the linear system of equations obtained from the parameterized Gosper
equation~\eqref{eq:2} by comparing coefficients with respect to both $n$
and~$k$. The corresponding matrix is much larger but its entries are integers
instead of integer polynomials.

As the resulting bound turns out to be much larger than the bound obtained in
the previous section for the height of the telescoper of order~$\nu$, we confine
ourselves to giving only an asymptotic estimate rather than an exact formula.
This makes the expressions in the calculations a little simpler.

Choose $r=2\nu=\OO(\Omega)$, $s=\delta+r\vartheta-\nu$,~$d=4\nu\vartheta=\OO(\Omega^2)$, and make an ansatz
\[
  L = \sum_{j=0}^r \sum_{i=0}^d \ell_{i,j} n^i S_n^j,\qquad
  Y = \sum_{j=0}^{s+d} \sum_{i=0}^s y_{i,j} n^i \binom kj,
\]
with undetermined coefficients $\ell_{i,j}$ and $y_{i,j}$. 
Then, comparing like coefficients of $n^i\binom kj$ in the equation
\[
  \sum_{j=0}^r \sum_{i=0}^d \ell_{i,j} n^i P_j
  = Q\,S_k(Y) - R\,Y
\]
leads to a system of homogeneous linear equations with 
\[
  (r+1)(d+1) + (s+d+1)(s+1) = 12\nu^2\vartheta^2 + (12 + 8\delta)\nu\vartheta + \nu^2 + \OO(\Omega)=\OO(\Omega^4)
\]
variables $\ell_{i,j}$ and $y_{i,j}$ and no more than
\begin{alignat*}1
    &\max\Big\{
    (\delta + r\vartheta + d + 1)(\delta + r\vartheta + 1), (\nu+s+d+1)(\nu + s + 1)\Bigr\}\\
    &= (\delta + r\vartheta + d + 1)(\delta + r\vartheta + 1) = 12\nu^2\vartheta^2
  +(8+8\delta)\nu\vartheta + \OO(1) = \OO(\Omega^4)
\end{alignat*}
equations. As $12>8$, this system has a nontrivial solution if $\nu\vartheta \to \infty$, as $\Omega \to \infty$. 

Let $A=(A_L,A_C)$ be the matrix encoding this linear system, with $A_L$ the
submatrix consisting of the columns corresponding to the variables $\ell_{i,j}$
and $A_C$ the part consisting of the columns corresponding to the variables
$y_{i,j}$, respectively. As the coefficients of $P_i$, $Q$, or $R$ do not change
when these polynomials are multiplied by some term~$n^j$ (only the exponents
change), we can use the same bounds for the heights of the matrix entries as
before. Hence $A_L$ is an integer matrix with $(r+1)(d+1)=\OO(\Omega^3)$ columns
and $\OO(\Omega^4)$ rows whose entries are bounded in absolute value by
$\e^{\OO(\Omega^2\log(\Omega))}$, and $A_C$ is an integer matrix with
$\OO(\Omega^4)$ rows and columns whose entries are bounded in absolute value by
$\e^{\OO(\Omega\log(\Omega))}$.

If $A_L$ happens to have full rank, we can apply Lemma~\ref{prop: matrix degree/height bound} 
to~$A$, interpreting its entries as integer polynomials
of degree zero.  It follows that the solution space has a basis whose components
are bounded by
\[
  \OO(\Omega^4)! (\e^{\OO(\Omega^2\log(\Omega))})^{\OO(\Omega^3)} (\e^{\OO(\Omega\log(\Omega))})^{\OO(\Omega^4)}
  = \e^{\OO(\Omega^5\log(\Omega))}.
\]
If $A_L$ does not have full rank, then, as before, any nontrivial solution of
$A_L$ gives rise to a nontrivial solution of $A$ by padding the solution vectors
with zeros. Applying Lemma~\ref{prop: matrix degree/height bound} to an arbitrary 
decomposition of $A_L$ into a block of full rank and the rest gives the bound
\[
  \OO(\Omega^3)! (\e^{\OO(\Omega^2\log(\Omega))})^{\OO(\Omega^3)} = \e^{\OO(\Omega^5\log(\Omega))}
\]
for the size of the integers in a basis of the solution space of~$A_L$. 
We have thus completed the proof of the following theorem. 

\begin{theorem}\label{thm:nonminimal}
  For every $\Omega\in\set N$, let $h_\Omega$ be a proper hypergeometric term as
  in \eqref{eq:hgdef}  for which the integer coefficients appearing in the $\Gamma$
  terms are bounded in absolute value by $\Omega$, for which $p$, $x$ and
  $y$ are fixed, and for which $\nu\vartheta \to\infty$ as $\Omega\to\infty$.
  Then, as $\Omega$ approaches infinity, each term $h_\Omega$
  admits a telescoper $L_\Omega$ of order $\OO(\Omega)$ and polynomial degree
  $\OO(\Omega^2)$ with integer coefficients bounded in absolute value
  by~$\e^{\OO(\Omega^5\log(\Omega))}$. 
\end{theorem}

There is nothing special about the choice $r=2\nu$ in the above derivation. The argument
works more generally for any choice $r=\alpha\nu$ where $\alpha>1$ is a constant (assumed
to remain fixed as $\Omega$ grows). Choosing $d=\frac{1+2\alpha}{\alpha-1}\nu\vartheta$
also leads to the bound $\e^{\OO(\Omega^5\log(\Omega))}$. 

For the (generically) minimal order $r=\nu$, the approach of this section only
delivers the height bound $\e^{\OO(\Omega^6\log(\Omega))}$ for a telescoper of
degree~$\OO(\Omega^3)$, which is much worse than the height bound
$\e^{\OO(\Omega^3\log(\Omega))}$ obtained in Theorem~\ref{thm:minimal} for a
telescoper of degree at most~$\OO(\Omega^4)$.

To conclude the section, we again compare the theoretical bound with the actual heights
found on a particular example. 

\begin{example}
  For $\Omega=1,2,3,\dots$ consider the same proper hypergeometric term
  $h_\Omega=\frac{\Gamma(\Omega k)}{\Gamma(\Omega n-k)}$ as in Example~\ref{ex:minimal}.
  From the minimal telescopers $L_\Omega$ of order $\nu=\Omega+1$, we constructed
  nonminimal telescopers of order $2\Omega$ of small degree and height.

  For each $L_\Omega$, we computed many terms of a randomly chosen sequence
  solution, and used these to construct a candidate operator $M_\Omega$ of order $2\Omega$
  and minimal degree by guessing. Checking that the $M_\Omega$ are
  left-multiples of the $L_\Omega$ proves that they are indeed
  telescopers. Unlike the minimal order operators $L_\Omega$, the minimal degree
  operators of order $2\Omega$ are typically not unique but form a vector space
  over $\set Q$ of dimension greater than~$1$. For example, for $\Omega=6$, the
  telescopers of order $12$ and degree $53$ form a vector space of dimension~$3$
  and there are no telescopers of order $12$ and degree $52$ or less. Using \edit{lattice 
  reduction~\citep{vzgathen99,nguyen10},}
  we determined an element of these vector spaces with small (but not
  necessarily smallest possible) integer coefficients. Let $H_\Omega$ be the logarithm of the
  maximum of the absolute values of the coefficients of the vector computed in
  this way.

  In Figure~\ref{fig:fitless}, we plot the values of $\frac{H_\Omega}{\Omega^5}$ (bullets, \tikz[baseline=-1ex] \node[fill,circle, inner sep=1pt] {};)
  against the least square fits
  \begin{enumerate}
  \item $\displaystyle\log(\Omega)\Bigl(0.269 + \frac{0.599}\Omega\Bigr)$ (solid, \tikz[baseline=-1ex] \draw[solid] (0,0)--(1,0);)
    \rule{0pt}{2em}
  \item $\displaystyle\frac{\log(\Omega)}{\Omega}\Bigl(2.73 - \frac{3.39}\Omega\Bigr)$ (dashed, \tikz[baseline=-1ex] \draw[dashed] (0,0)--(1,0);)
  \end{enumerate}
  for\rule{0pt}{1.7em} comparing the hypotheses $H_\Omega=\Theta(\Omega^5\log(\Omega))$ or
  $H_\Omega=\Theta(\Omega^4\log(\Omega))$.  Unfortunately, because of the high
  computational cost of computing~$H_\Omega$, we were not able to produce more
  data points. However, despite being less convincing than the test in the
  previous example, also here the solid curve seems to catch the trend better
  than the dashed curve, suggesting that the (quasi\hbox{-})quintic bound can
  probably not be improved to a (quasi\hbox{-})quartic bound in general.
  \edit{It also seems that the resulting bit size estimate $\OO(\Omega^8\log(\Omega))$
    is reasonably tight.}


\begin{figure}[th]
\centering\small
\begin{tikzpicture}[xscale=0.6,yscale=4]
   \draw[->] (0,0) -- (9,0) node[below]{$\Omega$};
   \draw[->] (0,0) -- (0,1) node[left]{$H_\Omega/\Omega^5$};
   \foreach \i in {1,2,...,8}
      \draw[very thin,yscale=0.15] (\i,0.2) -- +(0,-0.4);
   \foreach \i in {2,4,6,8}
      \node[below] at (\i,-0.02) {$\i$};
   \foreach \i in {0.2,0.4,0.6,0.8}
      \draw[xscale=0.8] (0.2,\i) -- +(-0.4,0) node[left]{$\i$};
   \draw[solid] plot coordinates{ 
      (1.000276, 0.000000) (1.006415, 0.005528) (1.043450, 0.035871) 
      (1.083448, 0.065896) (1.124770, 0.094280) (1.165196, 0.119768) 
      (1.204310, 0.142527) (1.247739, 0.165853) (1.286762, 0.185257) 
      (1.328428, 0.204523) (1.366183, 0.220815) (1.407459, 0.237483) 
      (1.446297, 0.252178) (1.486894, 0.266608) (1.526585, 0.279880) 
      (1.568137, 0.292968) (1.608157, 0.304860) (1.649084, 0.316360) 
      (1.689671, 0.327158) (1.726966, 0.336591) (1.769711, 0.346869) 
      (1.807943, 0.355616) (1.848706, 0.364514) (1.887721, 0.372645) 
      (1.931042, 0.381266) (1.968556, 0.388408) (2.011133, 0.396176) 
      (2.049938, 0.402963) (2.092385, 0.410090) (2.129009, 0.416005) 
      (2.170720, 0.422494) (2.211010, 0.428527) (2.251275, 0.434340) 
      (2.291391, 0.439929) (2.329930, 0.445120) (2.371591, 0.450544) 
      (2.411298, 0.455544) (2.453087, 0.460638) (2.490916, 0.465108) 
      (2.532725, 0.469900) (2.572770, 0.474352) (2.612718, 0.478665) 
      (2.654444, 0.483042) (2.692875, 0.486963) (2.732233, 0.490873) 
      (2.775689, 0.495074) (2.815036, 0.498777) (2.855275, 0.502469) 
      (2.896196, 0.506130) (2.933812, 0.509417) (2.973908, 0.512840) 
      (3.013700, 0.516159) (3.056264, 0.519627) (3.093825, 0.522620) 
      (3.137163, 0.525997) (3.176272, 0.528979) (3.214968, 0.531869) 
      (3.256579, 0.534914) (3.298359, 0.537907) (3.336372, 0.540578) 
      (3.376740, 0.543360) (3.416135, 0.546025) (3.458840, 0.548859) 
      (3.495792, 0.551267) (3.538232, 0.553984) (3.578129, 0.556492) 
      (3.617639, 0.558934) (3.657182, 0.561337) (3.697512, 0.563747) 
      (3.740262, 0.566259) (3.779459, 0.568525) (3.818152, 0.570727) 
      (3.859514, 0.573045) (3.900786, 0.575322) (3.937830, 0.577336) 
      (3.981711, 0.579687) (4.018429, 0.581625) (4.061375, 0.583860) 
      (4.102362, 0.585962) (4.138959, 0.587814) (4.180116, 0.589870) 
      (4.221546, 0.591911) (4.262779, 0.593915) (4.301007, 0.595749) 
      (4.340591, 0.597624) (4.381527, 0.599540) (4.422333, 0.601426) 
      (4.464307, 0.603341) (4.501277, 0.605008) (4.542897, 0.606863) 
      (4.584688, 0.608704) (4.624961, 0.610457) (4.661533, 0.612032) 
      (4.705020, 0.613883) (4.741860, 0.615435) (4.784713, 0.617220) 
      (4.822646, 0.618784) (4.864263, 0.620481) (4.903893, 0.622081) 
      (4.945243, 0.623733) (4.983215, 0.625234) (5.024173, 0.626838) 
      (5.066717, 0.628487) (5.103751, 0.629909) (5.143750, 0.631430) 
      (5.185072, 0.632987) (5.225497, 0.634496) (5.264611, 0.635942) 
      (5.308041, 0.637533) (5.347064, 0.638949) (5.388729, 0.640447) 
      (5.426484, 0.641793) (5.467761, 0.643253) (5.506599, 0.644614) 
      (5.547195, 0.646025) (5.586886, 0.647393) (5.628439, 0.648813) 
      (5.668459, 0.650169) (5.709385, 0.651545) (5.749972, 0.652898) 
      (5.787268, 0.654132) (5.830013, 0.655535) (5.868244, 0.656780) 
      (5.909007, 0.658098) (5.948023, 0.659350) (5.991344, 0.660729) 
      (6.028857, 0.661915) (6.071435, 0.663250) (6.110240, 0.664458) 
      (6.152686, 0.665770) (6.189311, 0.666895) (6.231021, 0.668166) 
      (6.271312, 0.669385) (6.311576, 0.670596) (6.351692, 0.671793) 
      (6.390231, 0.672936) (6.431893, 0.674163) (6.471599, 0.675324) 
      (6.513389, 0.676539) (6.551218, 0.677631) (6.593026, 0.678830) 
      (6.633072, 0.679971) (6.673020, 0.681103) (6.714745, 0.682277) 
      (6.753176, 0.683351) (6.792535, 0.684445) (6.835991, 0.685646) 
      (6.875337, 0.686726) (6.915576, 0.687823) (6.956498, 0.688933) 
      (6.994113, 0.689947) (7.034210, 0.691022) (7.074001, 0.692083) 
      (7.116565, 0.693211) (7.154127, 0.694200) (7.197464, 0.695335) 
      (7.236574, 0.696354) (7.275270, 0.697356) (7.316881, 0.698428) 
      (7.358660, 0.699498) (7.396673, 0.700466) (7.437041, 0.701489) 
      (7.476437, 0.702482) (7.519142, 0.703552) (7.556094, 0.704474) 
      (7.598534, 0.705526) (7.638430, 0.706511) (7.677941, 0.707480) 
      (7.717483, 0.708446) (7.757814, 0.709426) (7.800564, 0.710459) 
      (7.839761, 0.711402) (7.878454, 0.712328) (7.919816, 0.713313) 
      (7.961088, 0.714291) (8.000000, 0.715208) }
      node[right]{$\Theta(\Omega^5\log(\Omega))$};
   \draw[dashed] plot coordinates{ 
      (0.709702, 1.000000) (0.724412, 0.864877) (0.762344, 0.608900) 
      (0.803962, 0.401839) (0.843592, 0.258596) (0.884942, 0.151236) 
      (0.922914, 0.081478) (0.963872, 0.029831) (1.006415, -0.004021) 
      (1.043450, -0.020923) (1.083448, -0.029105) (1.124770, -0.029119) 
      (1.165196, -0.022837) (1.204310, -0.012288) (1.247739, 0.003250) 
      (1.286762, 0.019725) (1.328428, 0.039177) (1.366183, 0.057951) 
      (1.407459, 0.079280) (1.446297, 0.099788) (1.486894, 0.121410) 
      (1.526585, 0.142524) (1.568137, 0.164427) (1.608157, 0.185201) 
      (1.649084, 0.206019) (1.689671, 0.226164) (1.726966, 0.244190) 
      (1.769711, 0.264241) (1.807943, 0.281599) (1.848706, 0.299493) 
      (1.887721, 0.316020) (1.931042, 0.333685) (1.968556, 0.348400) 
      (2.011133, 0.364457) (2.049938, 0.378505) (2.092385, 0.393245) 
      (2.129009, 0.405451) (2.170720, 0.418789) (2.211010, 0.431120) 
      (2.251275, 0.442919) (2.291391, 0.454171) (2.329930, 0.464522) 
      (2.371591, 0.475226) (2.411298, 0.484975) (2.453087, 0.494775) 
      (2.490916, 0.503256) (2.532725, 0.512214) (2.572770, 0.520402) 
      (2.612718, 0.528202) (2.654444, 0.535973) (2.692875, 0.542804) 
      (2.732233, 0.549488) (2.775689, 0.556518) (2.815036, 0.562578) 
      (2.855275, 0.568488) (2.896196, 0.574213) (2.933812, 0.579233) 
      (2.973908, 0.584336) (3.013700, 0.589159) (3.056264, 0.594063) 
      (3.093825, 0.598181) (3.137163, 0.602697) (3.176272, 0.606565) 
      (3.214968, 0.610207) (3.256579, 0.613925) (3.298359, 0.617462) 
      (3.336372, 0.620515) (3.376740, 0.623591) (3.416135, 0.626436) 
      (3.458840, 0.629352) (3.495792, 0.631739) (3.538232, 0.634331) 
      (3.578129, 0.636629) (3.617639, 0.638775) (3.657182, 0.640801) 
      (3.697512, 0.642746) (3.740262, 0.644680) (3.779459, 0.646341) 
      (3.818152, 0.647880) (3.859514, 0.649419) (3.900786, 0.650849) 
      (3.937830, 0.652048) (3.981711, 0.653366) (4.018429, 0.654388) 
      (4.061375, 0.655494) (4.102362, 0.656462) (4.138959, 0.657258) 
      (4.180116, 0.658078) (4.221546, 0.658826) (4.262779, 0.659497) 
      (4.301007, 0.660056) (4.340591, 0.660574) (4.381527, 0.661046) 
      (4.422333, 0.661454) (4.464307, 0.661813) (4.501277, 0.662079) 
      (4.542897, 0.662324) (4.584688, 0.662515) (4.624961, 0.662649) 
      (4.661533, 0.662730) (4.705020, 0.662776) (4.741860, 0.662775) 
      (4.784713, 0.662730) (4.822646, 0.662651) (4.864263, 0.662524) 
      (4.903893, 0.662366) (4.945243, 0.662164) (4.983215, 0.661945) 
      (5.024173, 0.661676) (5.066717, 0.661361) (5.103751, 0.661058) 
      (5.143750, 0.660703) (5.185072, 0.660306) (5.225497, 0.659889) 
      (5.264611, 0.659460) (5.308041, 0.658955) (5.347064, 0.658477) 
      (5.388729, 0.657942) (5.426484, 0.657436) (5.467761, 0.656861) 
      (5.506599, 0.656300) (5.547195, 0.655693) (5.586886, 0.655080) 
      (5.628439, 0.654419) (5.668459, 0.653764) (5.709385, 0.653076) 
      (5.749972, 0.652377) (5.787268, 0.651721) (5.830013, 0.650952) 
      (5.868244, 0.650251) (5.909007, 0.649489) (5.948023, 0.648746) 
      (5.991344, 0.647907) (6.028857, 0.647169) (6.071435, 0.646319) 
      (6.110240, 0.645532) (6.152686, 0.644661) (6.189311, 0.643899) 
      (6.231021, 0.643021) (6.271312, 0.642164) (6.311576, 0.641297) 
      (6.351692, 0.640425) (6.390231, 0.639580) (6.431893, 0.638657) 
      (6.471599, 0.637770) (6.513389, 0.636828) (6.551218, 0.635970) 
      (6.593026, 0.635014) (6.633072, 0.634092) (6.673020, 0.633166) 
      (6.714745, 0.632192) (6.753176, 0.631291) (6.792535, 0.630362) 
      (6.835991, 0.629332) (6.875337, 0.628394) (6.915576, 0.627430) 
      (6.956498, 0.626446) (6.994113, 0.625537) (7.034210, 0.624566) 
      (7.074001, 0.623598) (7.116565, 0.622558) (7.154127, 0.621639) 
      (7.197464, 0.620574) (7.236574, 0.619611) (7.275270, 0.618655) 
      (7.316881, 0.617625) (7.358660, 0.616588) (7.396673, 0.615643) 
      (7.437041, 0.614638) (7.476437, 0.613655) (7.519142, 0.612588) 
      (7.556094, 0.611663) (7.598534, 0.610600) (7.638430, 0.609599) 
      (7.677941, 0.608607) (7.717483, 0.607613) (7.757814, 0.606599) 
      (7.800564, 0.605523) (7.839761, 0.604536) (7.878454, 0.603561) 
      (7.919816, 0.602519) (7.961088, 0.601479) (8.000000, 0.600498) }
      node[right]{$\Theta(\Omega^4\log(\Omega))$};
   \foreach \i/\j in {1/0.000, 2/0.491, 3/0.402, 4/0.696, 5/0.392, 6/0.718,
       7/0.826, 8/0.671}
      \node[fill,circle, inner sep=1pt] at (\i,\j) {};
\end{tikzpicture}
\caption{Heights of nonminimal telescopers}
\label{fig:fitless}
\end{figure}
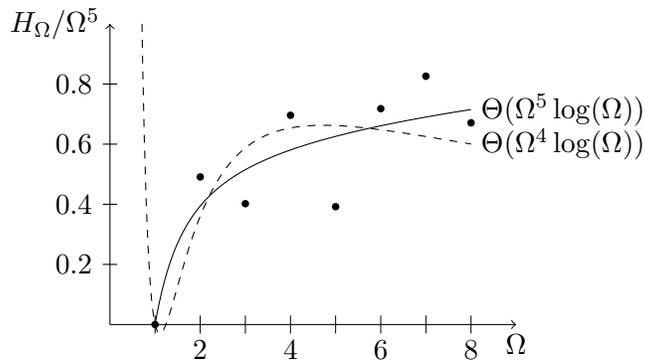
\end{example}

\section{Consequences}

Theorems~\ref{thm:minimal} and~\ref{thm:nonminimal} are primarily interesting for two reasons.
First, they give rise to a significant improvement of Yen's ``two-line
algorithm'' for proving hypergeometric summation identities~\citep{yen93,yen96}, and
second, they imply a bound on the bit complexity of creative telescoping.  No
such bound was known before.

The two-line algorithm rests on the following observations. 

\begin{proposition}[\citealt{yen93,yen96}]
  Let $L\in\set Z[n][S_n]$ be an operator of order~$r$ and degree~$d$,
  and let $\ell_r\in\set Z[n]\setminus\{0\}$ be the coefficient of $S_n^r$ in~$L$.
\begin{enumerate}
\item Suppose there is a sequence $(a_n)_{n=0}^\infty$ which is annihilated by
  $L$ and contains a run of at least $r+d+1$ consecutive $1$'s (i.e., there
  exists an index $n_0\in\set N$ with
  $a_{n_0}=a_{n_0+1}=\cdots=a_{n_0+r+d+1}=1$). Then $L$ also annihilates~$1$.
\item Let $(a_n)_{n=0}^\infty$ and $(b_n)_{n=0}^\infty$ be sequences which are
  annihilated by~$L$. If $a_n=b_n$ for all $n\leq r+n_0$, where $n_0$ is the
  greatest nonnegative integer root of $\ell_r$ (or $n_0=0$ if $\ell_r$ has no nonnegative integer
  roots). Then $a_n=b_n$ for all $n\in\set N$.
\item 
  If $n_0$ is an integer root of~$\ell_r$, then $n_0\leq \edit{|\ell_r|}$. 
\end{enumerate}
\end{proposition}

In view of these facts, in order to prove a hypergeometric summation identity
\[
  \sum_k h(n,k) = 1,
\]
for a given proper hypergeometric term $h(n,k)$ which has finite support and no singularities in
$\set N\times\set Z$, 
\edit{and for which also the term in \eqref{eq:Ch} has no singularity for any $r\in\set N$},
it suffices to proceed as follows:
\begin{enumerate}
\item Determine bounds on the order~$r$, the degree~$d$, and the height~$H$, of
  some telescoper of the summand~$h$.
\item Check the identity for $n=0,\dots,r+ \edit{H}$. It holds for all $n\in\set N$
  iff it holds for all these points.
\end{enumerate}
For step~1, Yen gives an explicit formula for a bound with asymptotic growth
$\e^{\OO(\Omega^6\log(\Omega))}$ ($\Omega\to\infty$). Our bound from
Theorem~\ref{thm:minimal} is significantly better, albeit still
exponential. Although, as illustrated in \edit{Example~\ref{ex:minimal}}, our bound seems to be
tight in general, it turns out that in virtually all examples the integer roots
of the leading coefficient~$\ell_r$ are much smaller than they could be. In
these cases, it remains much more efficient to compute a telescoper for the
summand and inspect the linear factors of~$\ell_r$.

For the cost of computing a telescoper, Theorem~8 of \citet{chen12c} says that a
telescoper of order~$r=\nu$ [~resp. $r=\OO(\Omega)$~] and degree $d=\OO(\Omega^3)$
[~resp. $d=\OO(\Omega^2)$~] can be computed using $\OO^\sim(\Omega^9)$
[~resp. $\OO^\sim(\Omega^8)$~] arithmetic operations, where the soft-O notation
$\OO^\sim(\cdot)$ suppresses possible logarithmic terms. If we use these
algorithms to compute telescopers modulo various primes and then use Chinese
remaindering and rational reconstruction to combine the results of the modular
computations into a telescoper with integer coefficients, this will take time
proportional to the length of the integers appearing in the output times the
number of arithmetic operations spent for a single prime. We thus obtain a bound
$\OO^\sim(\Omega^3)\times\OO^\sim(\Omega^9)=\OO^\sim(\Omega^{12})$ for the time to
compute a telescoper of order $r=\nu$ if no lower order telescoper exists, and a
bound of $\OO^\sim(\Omega^5)\times\OO^\sim(\Omega^8)=\OO^\sim(\Omega^{13})$ for the
time to compute a nonminimal telescoper of order~$r=\OO(\Omega)$.

There is another, somewhat more heuristic algorithm which makes use of the fact
that all the telescopers of a given term~$h$ form a left ideal in the operator
algebra~$\set Q(n)[S_n]$ (see~\citet{bronstein96} for a tutorial on arithmetic in
such algebras). The algorithm proceeds as follows. Choose a prime $p\in\set Z$
and compute several nonminimal telescopers, then take their greatest common
right divisor in $\set Z_p(n)[S_n]$, and hope that this is the modular image of
the minimal telescoper. With high probability, this will be the case. Repeat the
computation for various primes and use Chinese remaindering and rational
reconstruction to recover an operator in $\set Q(n)[S_n]$ from all the modular
greatest common right divisors. If we assume that the cost of computing the
greatest common right divisor can be neglected, then this algorithm spends
$\OO^\sim(\Omega^8)$ operations in $\set Z_p$ for every prime~$p$, and if we
further assume that possible issues related to unlucky primes can be neglected
as well, we expect to need $\OO^\sim(\Omega^3)$ primes of size $\OO^\sim(1)$. The
resulting bit complexity is thus
$\OO^\sim(\Omega^3)\times\OO^\sim(\Omega^8)=\OO^\sim(\Omega^{11})$ for terms $h$
whose minimal telescoper has order~$r=\nu$.

As pointed out above, for proving a hypergeometric identity it is not necessary
to explicitly compute a telescoper for the summand. Yen's algorithm gets away
without computing any information about the telescoper. It is however
very expensive. On the other hand, explicitly computing a complete telescoper is
more than we need, even though it is cheaper. The algorithm proposed
by~\citet{guo08} is an attempt to compromise between these two extremes: it
actually sets up the linear system for computing a telescoper, but then, instead
of solving it, it determines a bound on the height of the solution, taking into
account special features of the particular matrix at hand, such as sparsity, in
a more careful way than it would be easily possible to do in a general
analysis. Unfortunately, Guo et al.\ do not make any statement about the
complexity of their algorithm. It would be interesting to know whether their
improvement can be translated into better bounds on either the height of a
telescoper or, more generally, on the bit complexity of creative telescoping.

\section{Acknowledgement}

We would like to thank Mogens Lemvig Hansen for producing the graphs for both examples.

\bibliographystyle{plain}

\end{document}